\pdfoutput=1
\documentclass[12pt,draftcls,onecolumn,a4paper]{IEEEtran}
\usepackage{amsmath,amsfonts,amssymb}
\usepackage{graphicx}
\usepackage{cite}
\usepackage{subfigure}
\usepackage{multirow}
\usepackage{algorithm}
\usepackage{algorithmic}
\usepackage{mathtools}
\usepackage{hyperref}
\IEEEoverridecommandlockouts

\newtheorem{theorem}{Theorem}
\newtheorem{definition}{Definition}

\newtheorem{corollary}{Corollary}
\newtheorem{proposition}{Proposition}
\newtheorem{example}{Example}

\newcommand{\set}[1]{\mathcal{#1}}
\newcommand{\defined}{\triangleq}

\newcommand{\graph}{\set{G}}
\newcommand{\nodes}{\set{V}}
\newcommand{\edges}{\set{E}}
\newcommand{\N}{\set{N}}

\newcommand{\X}{\set{X}}

\newcommand{\A}{\set{A}}
\newcommand{\B}{\set{B}}

\newcommand{\eht}[2]{#1 \to #2}
\newcommand{\tail}[1]{\mathrm{tail}(#1)}
\newcommand{\head}[1]{\mathrm{head}(#1)}
\newcommand{\sessions}{\set{S}}

\def\PN{{\mathcal P( \mathcal N)}}
\def\support{{\text{Supp}
}}

\def\a{{a}}
\def\b{{b}}
\def\A{{X}^{*}}
\def\Ao{{X}}
\def\B{{Y}}
\def\iup{{i^{+}}}

\title{On the Capacity of Networks with Correlated Sources}
\author{\authorblockN{ Satyajit Thakor~\IEEEmembership{Member, IEEE},~Terence Chan~\IEEEmembership{Member, IEEE} and Alex Grant~\IEEEmembership{Senior Member,~IEEE}}\thanks{S. Thakor is with the
    Institute of Network Coding, Chinese University of Hong
    Kong. A. Grant and T. Chan are with the Institute for
    Telecommunications Research, University of South Australia. This
    work was performed in part while S. Thakor was with the
    Institute for Telecommunications Research, University of South
    Australia. T. Chan and A. Grant are supported in part by the
    Australian Research Council under Discovery Projects DP088022 and DP1094571.}}

\begin{document}
\maketitle

\begin{abstract}
  Characterizing the capacity region for a network can be extremely difficult. Even with independent sources, determining the capacity region can be as hard as the open problem of characterizing all information inequalities. The majority of computable outer bounds in the literature are relaxations of the Linear Programming bound which involves entropy functions of random variables related to the sources and link messages.  When sources are not independent, the problem is even more complicated.  Extension of linear programming bounds to networks with correlated sources is largely open. Source dependence is usually specified via a joint probability distribution, and one of the main challenges in extending linear programming bounds is the difficulty (or impossibility) of characterizing arbitrary dependencies via entropy functions. This paper tackles the problem by answering the question of how well entropy functions can characterize correlation among sources.  We show that by using carefully chosen auxiliary random variables, the characterization can be fairly ``accurate''. 
 \end{abstract}

\section{Introduction}

The fundamental question in network coding is to determine the
required link capacities to transmit the sources to the
sinks. Characterizing the network coding capacity region is extremely hard \cite{ChaGra08}.  When the sources are independent, the capacity region depends only on the source entropy rates. However,
when the sources are dependent, the capacity region depends on the
detailed structure of the joint source distribution.

Following \cite{Yeu02},  a linear programming outer bound was developed for dependent sources~\cite{ThaChaGra11} (see also \cite{ThaGraCha12}).  This bound is specified by a set of information inequalities and equalities, and source dependence is represented by the entropy function 
\begin{align}\label{1}
h(\alpha) \triangleq H(X_{s}^{i} , s\in\alpha), \alpha \subseteq \sessions
\end{align}
where $\sessions$ is an index set for the sources and $\{(X_{s}^{i},
s\in \sessions), i = 1, \ldots, K \}$ are $K$ independent and
identically distributed copies of the $|\sessions|$ dependent sources. Thus
each $(X_{s}^{i}, s\in\{1,\dots,S\})$ has the same joint distribution
as the sources, but are independent across different $i$.

However \eqref{1} fails to properly characterize source dependence.  
We also note that the capacity regions (or best known achievable
regions) for many classic multiterminal problems are also expressed as
optimizations of linear combinations of joint entropies, subject to
linear constraints (e.g. markov constraints) on joint entropies. If it
were not for the specified joint distributions on the
sources/side-information etc. typically present in such problems,
numerical solution would be achieved by a linear program. Again, if it
were possible to somehow accurately capture the dependence of random variables
using entropies, it would lead to a convenient computational approach.

A natural question arises: How accurately can arbitrary dependencies be specified via entropies alone? We will show that by using auxiliary random variables, entropies can in fact be sufficient.

%

\subsection*{Organization}
This work of characterizing correlation between random variables using entropy functions was mainly motivated by the problem of characterizing outer bounds on the capacity of networks with correlated sources. In Section \ref{sec:background} we review known outer bounds characterized using graph theoretic approach (referred as graphical bounds) as well as outer bounds using geometrical approach (referred as geometric bounds). These bounds are not tight and can be tightened by introducing new auxiliary random variables which more accurately describe correlation between the source random variables. In Section \ref{sec:Improved Outer Bounds}, we give a general framework for improving outer bounds with introduction of auxiliary random variables. In Section \ref{sec:ExampleCorrelated}, we demonstrate by an example that our LP bound can can in fact be tightened via the use of auxiliary random variables.  In Section \ref{sec:AuxCommonInfo} and Section \ref{sec:AuxLinearCorr}, we present two approaches to construct auxiliary random variables to tighten the outer bounds.  
The constructions via these two approaches are direct generalizations of the auxiliary random variables designed in Example \ref{ex:CorrelatedSources}, Section \ref{sec:ExampleCorrelated}. In Section \ref{sec:DistViaEnt}, we deal with the more general problem of characterizing probability distribution using entropy functions.

\section{Background}\label{sec:background}


Despite its importance, the maximal gain that can be obtained by network coding is still largely unknown, except in a few scenarios \cite{AhlCai00,YeuZha99}. One example is the single-source scenario where the capacity region is characterized by
the max-flow bound \cite{AhlCai00} (see also \cite[Chapter 18]{Yeu08}) and linear network codes maximize throughput \cite{LiYeu03}. However, when it involves more than one source, the problem can become quite difficult.

The problem becomes even more complex when the sources are correlated. In the classical literature, the problem of communicating correlated sources over a network is called distributed source compression. For networks of error-free channels with edge capacity constraints, the distributed source compression problem is a feasibility problem: given a network with edge capacity constraints and the joint probability distribution of correlated sources available at certain nodes, is it feasible to communicate the correlated sources to demanding nodes?

A relevant important problem is of separation of distributed source coding and network coding \cite{RamJaiChoEff06}. Specifically,  distributed source coding and network coding are separable if and only if optimality is not sacrificed  by separately designing source and network codes. It has been shown in \cite{RamJaiChoEff06} that the separation holds for two-source two-sink networks however it has been shown by examples that that the separation fails for two-source three-sink and three-source two-sink networks.

In this section, we present known outer bounds on the capacity of networks with correlated sources. We first describe network model and define network code and achievable rate. We then present known graphical and geometric outer bounds.

A network is modelled as a graph $\mathcal G = (\mathcal V, \mathcal E)$ where $\mathcal V$ is the set of nodes and $\mathcal E$ is the set of directed edges between certain pairs of nodes. Associated with each edge $e \in \mathcal E$ is a non-negative real number $C_e\geq0$ called the capacity of the edge $e$. For edges $e,f\in\edges$, we write $\eht{f}{e}$ as a shorthand for $\head{f}=\tail{e}$.  Similarly, for an edge $f\in\edges$ and a node $u\in\nodes$, the notations $\eht{f}{u}$ and $u\rightarrow f$ respectively denote $\head{f}=u$ and $\tail{f}=u$. Let $\sessions$ be an index set for a number of multicast sessions, and let $\{Y_{s}: s \in \sessions\}$ be the set of source variables. These sources are available at the nodes identified by the mapping 
\begin{equation}
a : \sessions \mapsto \nodes.
\end{equation}
Each source may be demanded by multiple sink nodes, identified by the mapping
 \begin{equation}
b : \sessions \mapsto 2^\mathcal V.
\end{equation}
where, $2^\mathcal V$
 is the set of all subsets of $\mathcal V$.
Each edge $e\in\edges$ carries a random variable $U_e$ which is a function of incident edge random variables and source random variables.

For a given network $\graph=(\nodes,\mathcal E)$ and
connection requirement $a$ and $b$, a network code is a set of mappings from input random variables (sources and incoming edges) to output random variables (outgoing
edges) at a network node. The mapping must obey constraints implied by the topology. The alphabets of source random variables $Y_{s}$ and edge random variables $U_{e}$ are denoted by $\mathcal Y_{s}, s \in \mathcal S$ and $\mathcal U_{e}, e \in \mathcal E$, respectively.

\begin{definition}[Network code]
A network code $(\Phi,\Psi)$ for a given network $\graph(\nodes,\mathcal E)$ is described by sets of its encoding functions $\Phi$ and decoding functions $\Psi$.

\begin{align}
\Phi &= \Biggl\{\phi_{e}: \prod_{s \in S : s \rightarrow e} \mathcal Y_{s} \times \prod_{f \in \mathcal E: \eht{f}{e}} \mathcal U_{f}
 \longmapsto \mathcal U_{e} , e \in \mathcal E \Biggr\}\\
\Psi &= \Biggl\{\psi_{u}: \prod_{f \in \mathcal E: \eht{f}{u}} \mathcal
U_{f} \longmapsto \mathcal Y_{s}, u \in b(s), s \in \mathcal S \Biggr\}
\end{align}
\end{definition}

Now we define an achievable rate tuple. The definition below is different from the usual definition of an achievable rate \cite[Definition 21.2]{Yeu08} in that the source rates are fixed and the link capacity constraints are variable.

\begin{definition}[Achievable rate tuple]
Consider a given network $\mathcal G=(\mathcal V,\mathcal E)$ with discrete memoryless sources   $\{Y_s, s\in {\mathcal S}\}$ and underlying probability distribution  $P_{Y_{\mathcal S}}(\cdot)$. 
%
A link capacity tuple $\mathbf{{C}} = (C_{e}: e \in \mathcal E)$ is called achievable if there exists a sequence of network codes $\phi_{\mathcal G}^{(n)}$ such that for every $e \in \mathcal E$ and every $s \in
\mathcal S$
\begin{align}
\lim_{n \rightarrow \infty} n^{-1} \log |\mathcal U^{(n)}_{e}| &\leq C_{e} \label{eq:sequenceofNC1}\\
\lim_{n \rightarrow \infty} \textrm{Pr}\{\psi_{u}^{(n)}(U^{(n)}_{f}: f \rightarrow u) \neq Y^{(n)}_s\} &= 0, \forall u \in b(s)\label{eq:sequenceofNC2}
\end{align}
where $\psi_{u}^{(n)}(U^{(n)}_{f}: f \rightarrow u)$ is the decoded estimate of $Y^{(n)}_s$ at node $u$ from $(U^{(n)}_{f}: f \rightarrow u)$ via mapping $\psi_{u}^{(n)}$.
\end{definition}

The set of all achievable link capacity tuples is denoted by $\mathcal R_{\text{cs}}$ where the subscript describes correlated source case.

\subsection{Graphical Bounds}
In \cite{Han11}, the author gave a necessary and sufficient condition for $\mathcal R_{\text{cs}}$ when each sink requires all the sources.\footnote{The results were generalized for networks with noisy channels. However, in this paper we are mainly concerned with networks with error-free channels.} This result includes the necessary and sufficient condition \cite{Han80}, \cite{BarSer06} for networks in which every source is demanded by single sink as a special case.

\begin{theorem}[Theorem 3.1, \cite{Han11}]\label{thm:Han bound}
For networks of error free channels, the transmission of sources $Y_s:s \in \mathcal W$ is feasible if and only if
\begin{equation}
H(Y_{\mathcal W}|Y_{\mathcal W^c}) \leq \min_{\mathcal T} \sum_{e:\substack{
 \head{e} \in \mathcal T,\\
 \tail{e} \in \mathcal T^{c}
}} C_{e}
\end{equation}
where source sessions $s: s \in \mathcal W$ are available at some nodes in $\mathcal T$ and all source sessions $s : s \in \mathcal W$ are demanded by at least one node in $\mathcal T^c$, i.e., this is the min-cut of the graph.
\end{theorem}

As mentioned above, for a few special cases a necessary and sufficient condition for reliable transmission of correlated sources over a network is given in \cite{BarSer06}, \cite{RamJaiChoEff06} and \cite{Han11}. However, the problem is an uncharted area in general. Until recently there did not even exist in the current literature a nontrivial necessary condition for reliable transmission of correlated sources in general multicast networks. In \cite{ThaGraCha12}, we made the first attempt to address this problem by characterizing a graph based bound, called the ``functional dependence bound'', for networks with correlated sources with arbitrary sink demands. The functional dependence bound was initially characterized for network with independent sources \cite{ThaGraCha09}. Later in \cite{ThaChaGra11}, we showed that the functional dependence bound is also an outer bound for networks with correlated sources.

In \cite{ThaGraCha12} we gave an abstract definition of a functional dependence graph, which expressed a set of local dependencies between random variables. In particular, we described a test for functional dependence, and gave a basic result relating local and global dependence. Below is the functional dependence bound based on the implications of functional dependence.

\begin{theorem}[Functional dependence bound \cite{ThaGraCha12}]\label{thm:mainresult1}
  Let $\graph=(\nodes,\edges)$ be a functional dependence
  graph 
   on the (source and edge) random variables
  $Y_\sessions,U_\edges$. Let $\boldsymbol{\mathcal M}$ be the
  collection of all maximal irreducible sets \cite[Definition 25]{ThaGraCha12}. Then
\begin{equation}
h(Y_{\mathcal W}|Y_{\mathcal W^c}) \leq \min_{\{U_{\mathcal A},Y_{\mathcal W^c}\} \in \boldsymbol{\mathcal M}} \sum_{e \in \mathcal A} C_e,
\end{equation}
where $\{ U_{\mathcal A},Y_{\mathcal W^c}\} \in \boldsymbol{\mathcal M}$ and $Y_{\mathcal W},Y_{\mathcal W^c} \subseteq Y_{\mathcal S}$. 
\end{theorem}

The \emph{functional dependence region} is defined as follows.
%

\begin{equation}\label{eq:R_FD}
\mathcal R_{FD} \triangleq \bigcap_{\mathcal W \subseteq \mathcal S} \left\{ \mathbf{h} \in \mathbb{R}_{+}^{2^{|\mathcal S|}-1}:h(Y_{\mathcal W}|Y_{\mathcal W^c}) \leq \min_{\{U_{\mathcal A},Y_{\mathcal W^c}\} \in \boldsymbol{\mathcal M}} \sum_{e \in \mathcal A} C_e\right\}
\end{equation}
where edge-sets $\mathcal A \subseteq \mathcal E$ are subsets of maximal irreducible sets.

We also generalized existing bounding techniques that characterize geometric bounds for multicast networks with independent sources for networks with correlated sources.

\subsection{Geometric Bounds}\label{sec:CorrelatedBound}
In this section, we focus on outer bounds on achievable rate region for networks with correlated sources using geometric approach. We present outer bounds $\mathcal R_{\text{cs}} (\overline{\Gamma^*})$ by using the set of almost entropic variables $\overline{\Gamma^*}$, and $\mathcal R_{\text{cs}}(\Gamma)$ (again called LP bound)
by using the set of polymatroid variables $\Gamma$ similar to the bounds given for independent sources in \cite[Chapter 15]{Yeu02}. 



\begin{definition}\label{def:regionsCS}
Consider a network coding problem for a set of correlated source random variables $\{Y_{s},s\in\mathcal S \}$. Let $\mathcal R_{\mathrm{cs}}(\Delta)$ be the set of all link capacity tuples $\mathbf{C} = (C_{e}: e \in \mathcal E)$
such that there exists a function $h \in \Delta$ (over the set of variables $\{X_{s},s\in \mathcal S, U_{e},e\in\mathcal E\}$) satisfying the following constraints:
\begin{align}
h (X_{\mathcal W}: \mathcal W \subseteq \mathcal S)&=H(Y_{\mathcal W}) \label{eq:R_out1}\\
h (U_{e}|\{X_{s}: a(s) \rightarrow e\},\{U_{f}: f \rightarrow e\}) &= 0\\
h (X_{s}: u \in b(s)| U_{e}: e \rightarrow u)) &=0, \forall u \in b(s)\\
h (U_{e}) &\leq C_{e} \label{eq:R_out5}
\end{align}
for all $s\in\mathcal S$ and $e\in \mathcal E$.
\end{definition}

Taking $\Delta$ as $\overline{\Gamma^*}$ and $\Gamma$ in Definition \ref{def:regionsCS} gives us regions $\mathcal R_{\mathrm{cs}}(\overline{\Gamma^*})$ and $\mathcal R_{\mathrm{cs}}(\Gamma)$ respectively.

\begin{theorem}[Outer bound \cite{ThaGraCha12}]\label{thm:R_Out}
$$\mathcal R_{\mathrm{cs}} \subset \mathcal R_{\mathrm{cs}}(\overline{\Gamma^*})$$
\end{theorem}

It is well known that the region $\overline{\Gamma^*}$ is closed and convex \cite{ZhaYeu97}. Moreover, the regions defined by the constraints \eqref{eq:R_out1}-\eqref{eq:R_out5} are also closed and convex. Replacing $\overline{\Gamma^*}$ by $\Gamma$ in Theorem
\ref{thm:R_Out}, we obtain an outer bound $\mathcal R_{\text{cs}}(\Gamma)$, for capacity of networks with correlated sources. This bound (a linear programming bound) is an outer bound for the achievable rate region since
$\overline{\Gamma^*} \subset \Gamma$ and Theorem \ref{thm:R_Out} implies
\begin{theorem}[Outer bound \cite{ThaGraCha12}]\label{thm:R_Out2}
\begin{equation}\label{eq:boundsCS}
\mathcal R_{\mathrm{cs}} \subset \mathcal R_{\mathrm{cs}}(\overline{\Gamma^*}) \subset \mathcal R_{\mathrm{cs}}(\Gamma).
\end{equation}
\end{theorem}

It is possible that the outer bounds $\mathcal R_{\text{cs}}(\overline{\Gamma^*})$ and $\mathcal R_{\text{cs}}(\Gamma)$ given above, in terms of the region of almost entropic vectors $\overline{\Gamma^*}$ and the region of polymatroid vectors $\Gamma$, may not be tight since the representation of
the regions $\overline{\Gamma^*}$ or $\Gamma$ together with constraints \eqref{eq:R_out1}-\eqref{eq:R_out5} do not capture the exact correlation of source random variables, i.e., the exact joint probability distribution. This is because the same entropy vector induced by
the correlated sources may be satisfied by more than one probability distribution. The importance of incorporating the knowledge of source correlation (joint distribution) to improve the cut-set bound is also recently and independently investigated in \cite{GohYanJagg11}.

\section{Improved Outer Bounds}\label{sec:Improved Outer Bounds}
In this section, we give a general framework for improved outer bounds using auxiliary random variables. In Section \ref{sec:ExampleCorrelated} we will demonstrate by an example that the outer bound $\mathcal R_{\text{cs}}(\Gamma)$ is not tight and also give an explicit improved outer bound which is strictly better than the outer bound $\mathcal R_{\text{cs}}(\Gamma)$. In Section \ref{sec:AuxCommonInfo} and \ref{sec:AuxLinearCorr}, we present two generalizations of Example \ref{ex:CorrelatedSources} to construct auxiliary random variables to obtain improved bounds.

\begin{definition}\label{def:Improved Outer Bound}
Consider a set of correlated sources $\{Y_s, s\in {\mathcal S}\}\div$ with underlying probability distribution $P_{Y_{\mathcal S}}(\cdot)$. Construct any auxiliary random variables
$K_{i},i\in \mathcal L$ by choosing a conditional probability distribution function
$P_{K_{\mathcal L}|Y_{\mathcal S}}(\cdot)$.

Let $\mathcal R'_{\mathrm{cs}}(\overline{\Gamma^*})$ be the set of all link capacity tuples $\mathbf{{C}} = (C_{e}: e \in \mathcal E)$ such that there exists an almost entropic function $ h  \in \overline{\Gamma^*}$ satisfying the following constraints:
\begin{align}
 h (X_{\mathcal W}: \mathcal W \subseteq \mathcal S, Z_{\mathcal Z}: \mathcal Z \subseteq \mathcal L)&=H(Y_{\mathcal W},K_{\mathcal Z})\\
 h (U_{e}|\{Y_{s}: a(s) \rightarrow e\},\{U_{f}: f \rightarrow e\}) &= 0\\
 h (Y_{s}: u \in b(s)| U_{e}: e \rightarrow u) &=0, \forall u \in b(s)\\
 h (U_{e}) &\leq C_{e}
\end{align}
for all $s\in\mathcal S$ and $e\in \mathcal E$.
\end{definition}

Similarly, an outer bound $\mathcal R'_{\text{cs}}(\Gamma)$ can be defined in terms of polymatroid function $h \in \Gamma$.

\begin{theorem}[Improved Outer bounds]
\begin{equation}
\mathcal R_{\mathrm{cs}} \subseteq
\mathcal R'_{\mathrm{cs}}(\overline{\Gamma^*}) \subseteq \mathcal R_{\mathrm{cs}}(\overline{\Gamma^*})\subseteq \mathcal R_{\mathrm{cs}}(\Gamma)
\end{equation}
\begin{equation}
\mathcal R_{\mathrm{cs}} \subseteq \mathcal R'_{\mathrm{cs}}(\Gamma) \subseteq \mathcal R_{\mathrm{cs}}(\Gamma)
\end{equation}
\end{theorem}

An improved functional dependence bound can also be obtained from the functional dependence bound by introducing auxiliary random variables. The improvement of the bounds of the form in Definition \ref{def:Improved Outer Bound} over the bound without using auxiliary random variables solely depends on the construction of auxiliary random variables.



\subsection{Looseness of the Outer Bounds}\label{sec:ExampleCorrelated}
In this section, we demonstrate by an example that
\begin{enumerate}
\item the LP bound $\mathcal R_{\text{cs}}(\Gamma)$ presented in Section \ref{sec:CorrelatedBound} is in fact loose and
\item the bounds derived in Section \ref{sec:CorrelatedBound} can be tightened by introducing auxiliary random variables.
\end{enumerate}
\begin{example}\label{ex:CorrelatedSources}
In Figure \ref{fig:net-cs}, three correlated sources $Y_1, Y_2, Y_3$ are available at node 1 and are demanded at nodes $3,4,5$ respectively. The edges from node $2$ to nodes $3,4,5$ have sufficient capacity to carry the random
variable $U_1$ available at node 2. The correlated sources $Y_1, Y_2, Y_3$ are defined as follows.
\begin{align}
Y_1 = (b_0,b_1)\\
Y_2 = (b_0,b_2)\\
Y_3 = (b_1,b_2)
\end{align}
where $b_0,b_1,b_2$ are independent, uniform binary random variables.
\begin{figure}[htbp]
\centering
  \includegraphics[scale=0.5]{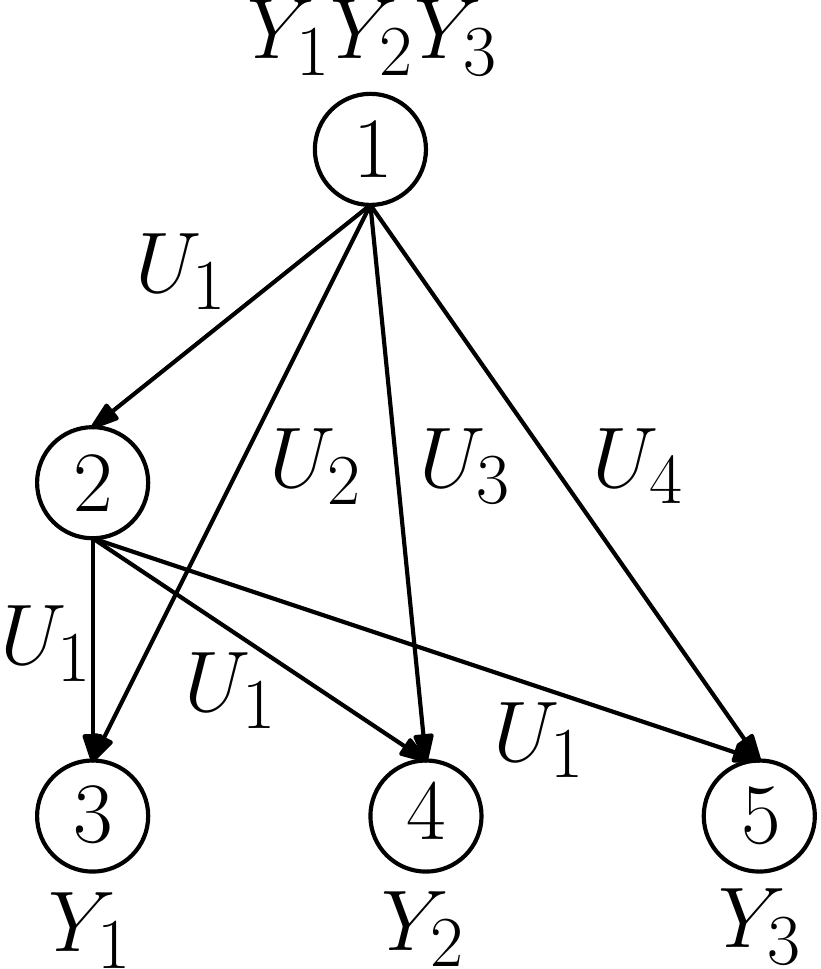}
  \caption{A network with correlated sources.}\label{fig:net-cs}
\end{figure}

\begin{definition}
The LP bound $\mathcal R_{\text{cs}}(\Gamma)$ for the network in Figure \ref{fig:net-cs} is the set of all link capacity tuples $\textbf{\textrm{C}} =(C_e, e=1,...,4)$ such that there exists $ h \in \Gamma$  satisfying the following constraints.
\begin{align}
 h (X_s)&=2, s=1,2,3 \label{eq:ex-cs1}\\
 h (X_i,X_j) &=3, i \neq j, i,j \in \{1,2,3\}\\
 h (U_e|X_1,X_2,X_3)&=0, e=1,2,3,4\\
 h (X_1|U_2U_1)&=0\\
 h (X_2|U_3U_1)&=0\\
 h (X_3|U_4U_1)&=0\\
 h (U_e)  &\leq C_e, i=1,...,4 \label{eq:ex-cs8}
\end{align}
\end{definition}

Note that the link capacity tuple $(C_e=1, e=1,...,4)$ is in the region $\mathcal R_{\text{cs}}(\Gamma)$ 
by choosing $h$ as the entropy function of the following random variables:
%
\begin{align*}
X_1 &= (b_0,b_1),& U_{1} & =b_{0}\\
X_2 &= (b_0,b_2) & U_{2} & = b_{1}\\
X_3 &= (b_0,b_1 \oplus b_2) & U_{3} & = b_{2} \\
& & U_{4} & = b_{1} \oplus b_{2}
\end{align*}
which satisfies \eqref{eq:ex-cs1}-\eqref{eq:ex-cs8} and polymatroidal axioms since these are random variables.

%

\def\bfh{{h}}

Now, we will characterize an improved LP bound by constructing auxiliary random variables $Z_{0},Z_{1},Z_{2}$.
\begin{definition}\label{def:ImrovedLPnet-cs}
An improved LP bound $\mathcal R'_{\text{cs}}(\Gamma)$ for the network in Figure \ref{fig:net-cs} is the set of all link capacity tuples $\textbf{\textrm{C}} =(C_e, e=1,...,4)$ such that there exists $ h \in \Gamma$  satisfying the following constraints.
\begin{align}
 h (X_s)&=2, s=1,2,3  \label{eq:iLPs}\\
 h (X_i,X_j) &=3, i \neq j, i,j \in \{1,2,3\}\\
 h (Z_\alpha)&=|\alpha|, \alpha \subseteq \{0,1,2\}  \\
\bfh({X_{1}|Z_{0},Z_{1}}) &=0 \\
\bfh({X_{2}|Z_{0},Z_{2}}) &=0 \\
\bfh({X_{3}|Z_{0},Z_{3}}) &=0 \\
\bfh(Z_0Z_1)&=h(X_1)\\
\bfh(Z_0Z_2)&=h(X_2)\\
\bfh(Z_1Z_2)&=h(X_3)\\
 h (U_e|X_1,X_2,X_3)&=0, e=1,2,3,4\\
 h (X_1|U_2U_1)&=0\\
 h (X_2|U_3U_1)&=0\\
 h (X_3|U_4U_1)&=0\\
 h (U_e)  &\leq C_e, i=1,...,4.  \label{eq:iLPe}
\end{align}
\end{definition}

Note that, by Definition \ref{def:ImrovedLPnet-cs}, the link capacity tuple $\textbf{\textrm{C}} =(C_e,
e=1,...,4)$ is in the improved LP bound $\mathcal R'_{\text{cs}}(\Gamma)$ if and only if there exists a  polymatroidal $h$  satisfying \eqref{eq:iLPs}-\eqref{eq:iLPe}. In the following, we prove that the link capacity tuple $(C_e=1, e=1,...,4)$ is indeed not in $\mathcal R'_{\text{cs}}(\Gamma)$, Definition \ref{def:ImrovedLPnet-cs}, and hence is not achievable.

Suppose to the contrary that $(C_e=1, e=1,...,4)$ is in $\mathcal R'_{\text{cs}}(\Gamma)$. Then by definition, there exists  a polymatroid $h$  satisfying \eqref{eq:iLPs}-\eqref{eq:iLPe}.
From these constraints, it is easy to prove that
\begin{align}
\bfh(U_{1}|Z_0Z_1)&=0\\
\bfh(U_{1}|Z_0Z_2)&=0\\
\bfh(U_{1}|Z_1Z_2)&=0 \\
\bfh(Z_{0}Z_{1}Z_{2}) & = \bfh(Z_{0}) + \bfh(Z_{1}) + \bfh(Z_{1}). \label{eq:4}
\end{align}
As $\bfh(U_{1}|Z_0Z_2) = 0$, it implies that
\begin{equation}
I_{\bfh}(U_{1};Z_{1}|Z_0Z_2) = 0.
\end{equation}

On the other hand, by \eqref{eq:4}, we have

\begin{equation}
I_{\bfh}(Z_{1} ; Z_{2} |Z_{0} ) = 0.
\end{equation}

Therefore,
\begin{equation}
I_{\bfh}(Z_{1} ; Z_{2}, U_{1} | Z_{0} ) = 0
\end{equation}
and consequently,
\begin{equation}
I_{\bfh}(Z_{1} ; U_{1} | Z_{0} ) = 0.
\end{equation}

Together with $\bfh(U_{1} | Z_{0},Z_{1}) = 0$, this implies
$\bfh(U_{1} | Z_{0}) = 0$. Similarly, we can also prove
that
\begin{equation}
\bfh(U_{1} | Z_{2})  = \bfh(U_{1} | Z_{0}) = 0.
\end{equation}

Using the same  argument,  we can once again prove that
$\bfh(U_{1}|Z_{1}) = \bfh(U_{1}|Z_{2}) = 0$ and $\bfh(Z_{1}Z_{2}) = \bfh(Z_{0}) + \bfh(Z_{1})$
implies $H(U_{1}) = 0$.

Finally,
$\bfh(X_{1}  | U_{1},U_{2}) = 0$ implies
\begin{align}
2 = \bfh(X_{1}) & \le \bfh( U_{1},U_{2}) \\
& \le \bfh( U_{1}) + \bfh( U_{2}) \\
& = \bfh( U_{2})  \\
& \le 1.
\end{align}

A contradiction occurs and hence, there exists no such polymatroidal $\bfh$
which satisfies \eqref{eq:iLPs}-\eqref{eq:iLPe}.
 In other words,
the  link capacity tuple $(C_e=1, i=1,...,4)$ is not in $\mathcal R'_{\text{cs}}(\Gamma)$, Definition \ref{def:ImrovedLPnet-cs}. (End of Example 1)
\end{example}


In Definition \ref{def:Improved Outer Bound}, we present new improved outer bounds on the capacity region of networks with correlated sources using auxiliary random variables. However, there is one problem that remains to be solved: How to construct auxiliary random variables that can tighten the bounds or more generally, can lead to characterization of the capacity region for networks with correlated sources. While it appears to be a hard problem to answer in general, we propose three approaches to construct auxiliary random variables. First, we propose to construct auxiliary random variables from common information. 

\subsection{Auxiliary Random Variables from Common Information}\label{sec:AuxCommonInfo}


The first approach is to construct an auxiliary random variable which is almost the common information of two random variables. This approach is a direct generalization of Example \ref{ex:CorrelatedSources} in the previous section in a sense that the auxiliary random variables in Example \ref{ex:CorrelatedSources} are precisely the common information between pairs of source random variables. This fact also implies that the approach leads to characterization of improved bounds.

\begin{definition}[Common Information \cite{GacKor73}]
For any random variables $X$ and $Y$, the \emph{common information} of $X$ and $Y$ is the  random variable (denoted by $C(X,Y)$) which has the maximal entropy among all other random variables $K$ such that
\begin{align}
    H(K|X) &=0\\
    H(K|Y) &=0.
\end{align}
%
%
\end{definition}

In many cases, it is not easy to find the common information between two random variables. For example, let $Z$ be a binary random variable such that $\Pr(Z=0)=\epsilon >0$ and $\Pr(Z=1) = 1-\epsilon$. Suppose $X$ is another binary
random variable independent of $Z$ and $Y=X \oplus Z$. Then if \cite{GacKor73} ( see also \cite{Wyn73})
\begin{align}
    H(K|X) &=0\\
    H(K|Y) &=0,
\end{align}
then $H(K)=0$ even if $X$ and $Y$ are almost the same for sufficiently small $\epsilon$. 

To address this issue, we propose a different way to construct auxiliary random variables. Consider any pair of random variables $\{X,Y\}$ with probability distribution $ P_{XY}(\cdot)$. For any $\delta \ge 0$, let
\begin{equation}
  \mathcal P(\delta) \triangleq \left\{ P_{K|XY}(\cdot):
  \begin{array}{l l}
     H(K|X) \leq \delta,\\
    H(K|Y) \leq \delta,\\
    I(X;Y|K) \leq \delta
  \end{array} \right\}
\end{equation}
where the probability distribution of $\{X,Y,K\}$ is given by
\begin{equation}
\Pr(X=x,Y=y,K=k) \defined P_{XY}(X=x,Y=y)P_{K|XY}(K=k|X=x,Y=y).
\end{equation}
Note that the ``smaller'' the $\delta$ is, the more similar the random variable $K$ (associated with the conditional distribution  $P_{K|XY}$) is to the common information.

Our constructed random variable will be selected from $\mathcal P(\delta^{*})$ to formulate an improved LP bound where
\begin{equation}
\delta^{*} = \min_{\delta : \mathcal P(\delta) \neq \emptyset} \delta.
\end{equation}

For a multi-source multicast network with source random variables $Y_1,\ldots,Y_{|\mathcal S|}$ one can construct random variables $K_{ij}, i \neq j, i,j \in \mathcal S$ from the family of distributions
\begin{equation}
  \mathcal P(\delta) \triangleq \left\{ P_{K_{ij}|Y_i,Y_j}(\cdot):
  \begin{array}{l l}
     H(K_{ij}|Y_i) \leq \delta,\\
    H(K_{ij}|Y_j) \leq \delta,\\
    I(Y_i;Y_j|K_{ij}) \leq \delta
  \end{array} \right\}.
\end{equation}

An improved LP bound for a multi-source multicast network with source random variables $Y_1,\ldots,Y_{|\mathcal S|}$ can be computed by constructing the random variables $K_{ij}, i \neq j, i,j \in \mathcal S$ and then taking inequalities
\begin{align}
     H(K_{ij}|Y_i) \leq \delta_{ij},\\
    H(K_{ij}|Y_j) \leq \delta_{ij},\\
    I(Y_i;Y_j|K_{ij}) \leq \delta_{ij}
\end{align}
into consideration apart from constraints \eqref{eq:R_out1}- \eqref{eq:R_out5} and elemental inequalities.




\subsection{Linearly Correlated Random Variables}\label{sec:AuxLinearCorr}
In some scenarios, source random variables are ``linearly correlated''. In this section we present a construction method for auxiliary random variables describing linear correlation between random variables. This approach is also a direct generalization of Example \ref{ex:CorrelatedSources} in the previous section in a sense that the source random variables are linearly correlated. 

\begin{definition}
A set of random variables $\{Y_{1}, \ldots, Y_{n}\}$ is called linearly correlated if
\begin{enumerate}
\item for any $\a \subseteq \{1,\ldots, n\}$, the support of the probability distribution of $(Y_{i},i\in\a)$ is a vector subspace and
\item $(Y_{\a}:\a \subseteq \{1,\ldots,n\})$ is uniformly distributed.
\end{enumerate}
\end{definition}

Let $\{Y_{1}, \ldots, Y_{n}\}$ be a set of linearly correlated random variables with support vector subspaces
\begin{equation}
V_i \subseteq \mathbb{F}^{m}_{q}
\end{equation}
and
\begin{equation}
\langle V_i : i \in n \rangle = \langle \B_i : i \in m \rangle \label{eq:SecLinear1}
\end{equation}
where $\B_i,i \in m$ are linearly independent. That is, $\{\B_i, i \in m\}$ is a basis for the subspaces $\{V_i : i \in n\}$. It can be noticed that, there exists a set of linearly independent random variables $K_{1}, \ldots, K_{m}$ uniformly distributed over the support $\{\B_i, i \in m\}$ induced from a basis of the vector spaces $\{V_i : i \in n\}$. That is,
\begin{equation}
H(K_{\a})=\a \log_2 m. \label{eq:SecLinear2}
\end{equation}
The random variable $Y_i$ can be written as a function of random variables $K_1,\ldots,K_m$ as follows
\begin{equation}
Y_i = [K_1 \ldots K_m] \mathbf{A}^i
\end{equation}
where 
\begin{equation}
\mathbf{A}^i=
 \begin{bmatrix}
  a^i_{1,1} & a^i_{1,2} & \cdots & a^i_{1,\mathrm{dim}(V_i)} \\
  a^i_{2,1} & a^i_{2,2} & \cdots & a^i_{2,\mathrm{dim}(V_i)} \\
  \vdots  & \vdots  & \ddots & \vdots  \\
  a^i_{m,1} & a^i_{m,2} & \cdots & a^i_{m,\mathrm{dim}(V_i)}
 \end{bmatrix}
\end{equation}
is an $m \times \mathrm{dim}(V_i)$ coefficient matrix. 

Thus, the random variables $Y_1,\ldots,Y_n$ are linear functions of the random variables $K_{1}, \ldots, K_{m}$. In particular, a random variable $Y_i$ is a function of the random variables $K_j$ such that the coefficient of $K_j$ is non-zero. Then we have the following equalities.
\begin{align}
H(Y_i|K_j: a^i_{jl}\neq0, \forall l \in \{1,\ldots, \mathrm{dim}(V_i)\},a^i_{jl}\in \mathbf{A}^i ) = 0\label{eq:SecLinear3}\\
H(Y_i)=H(K_j: a^i_{jl}\neq0, \forall l \in \{1,\ldots, \mathrm{dim}(V_i)\},a^i_{jl}\in \mathbf{A}^i )\label{eq:SecLinear4}
\end{align}

An improved LP bound can be computed by taking equalities \eqref{eq:SecLinear2},\eqref{eq:SecLinear3} and \eqref{eq:SecLinear4} into consideration apart from constraints \eqref{eq:R_out1}- \eqref{eq:R_out5} and elemental inequalities.

\section{Probability Distribution using Entropy Functions}\label{sec:DistViaEnt}

The basic question is: \emph{How ``accurate'' can entropy function specify the correlation among random variables}? We partly answer the question by 
showing that the joint probability distribution among random variables can be completely specified by entropy functions subject to some moderate constraint. First, we describe a few notations.

\subsubsection*{Notations}
Let $\mathcal N=\{1,\ldots,n\}$ and $X$ be a random variable.
Assume without loss of generality that  $X$ has a \emph{positive} probability distribution over $\N$.
Let $\PN$ be the set of all \emph{nonempty} subsets of $\{2,\ldots,n\}$.  
The size of the support\footnote{Roughly speaking, $\support(Y)$ is the number of possible values that $Y$ can take with positive probabilities.} of a random variable $Y$ will be denoted by $\support(Y)$. For notational simplicity, we will not distinguish a set with a single element $\{ i \}$ and the element $i$.   Two random variables are regarded as equivalent if they are functions of each other. Therefore, $\Ao$ and $\Ao + 1$ are regarded as equivalent.

Let 
\begin{align}
h_{b}(q) \triangleq -q \log q - (1-q) \log (1-q).
\end{align}
The function $h_{b}(q)$ is not one-to-one over the interval. Yet, we will use 
$h_{b}^{-1}(\delta)$ to define as the unique $q \in [0, 1/2] $ such that 
\begin{align}
h_{b}(q) = \delta.
\end{align}


\subsection{Single Random Variable Case}

First we consider the problem of characterizing distribution of single random variable via entropy functions. To understand the idea, consider a binary random variable $X$ such that $p_{X}(0) = p$ and 
$p_{X}(1) = 1-p$. While the entropy of $X$ does not determine exactly what the probabilities of $X$ are, it essentially determines the probability distribution (up to permutations).  To be precise, let $0\le q \le 1/2$  such that 
$
H(X) = h_{b}(q)
$
where
$
h_{b}(q) \triangleq -q\log q - (1-q) \log (1-q).
$
Then 
either $p=q$ or $p= 1-q$. Furthermore, the two possible distributions are in fact permutations of each other.

When $X$ is not binary, the entropy $H(X)$ alone is not sufficient to characterize the probability distribution of $X$. However, by using auxiliary random variables, it turns out that the distribution of $X$ can still be determined.  

The idea is best demonstrated by an example. Suppose $X$ is ternary, taking values from the set $\{1,2,3\}$. Suppose also that $p_{X}(x) > 0$ for all $x\in \{1,2,3\}$.
Define random variables $A_{1}$, $A_{2}$ and $A_{3}$ such that 
\begin{align}\label{2b}
A_{i} = 
\begin{cases}
1 & \text{ if } X = i \\
0 & \text{ otherwise. }
\end{cases}
\end{align}
Clearly, 
\begin{align}\label{2a}
H(A_{i}|X) = 0
\end{align}
 and 
\begin{align}\label{2}
H(A_{i}) = h_{b}(p_{X}(i)).
\end{align}

Let us further assume that $p_{X}(i) \le 1/2$ for all $i$. Then by \eqref{2} and strict monotonicity of 
$h_{b}(q)$ in the interval $[0, 1/2]$, it seems at the first glance that the distribution of $X$ is uniquely specified by the entropies of the auxiliary random variables. However, this is only half of the story and there is a catch in the argument -- The auxiliary random variables chosen are not arbitrary. When we ``compute'' the probabilities of $X$ from the entropies of the auxiliary random variables, it is assumed to know how the random variables are constructed. Without knowing the ``construction'',  it is unclear how to find the probabilities of $X$ from entropies.
More precisely, suppose we only know that there exists auxiliary random variables $A_{1},A_{2},A_{3}$ such that \eqref{2a} holds and their entropies are given by  
\eqref{2} (without knowing that the random variables are specified by \eqref{2b}). Then we cannot determine precisely what the distribution of $X$ is.  
Having said that, in this paper we will show that the distribution of $X$ can in fact be fully characterized by the ``joint entropies'' of the auxiliary random variables.

\subsubsection{Construction of auxiliary random variables}

\begin{definition}[Constructing auxiliary random variables $\A_{\a}$]\label{def:RVsAalpha}
%
For any $\a \in \PN$, let $\A_{\a}$ be the auxiliary random variables such that 
\begin{equation}
  \A_{\a} = \left\{
  \begin{array}{l l}
    1 & \quad \text{if $\A \in \a$}\\
    0 & \quad \text{otherwise}\\
  \end{array} \right.
\end{equation}
\end{definition}

Notice that $\A_{\a} = 0$ when $\A=1$.

\begin{proposition}[Property 1: Distinct]\label{prop:distinct}
For any distinct  $\a,  \b \in \PN$, then
\begin{align}
H(\A_\a|\A_{\b})&>0, \text{ and } \label{prop1:a}\\
H(\A_{\b}|\A_{\a})&>0. \label{prop1:b}
\end{align}
\end{proposition}

\begin{proof}
First note that $1 \in \mathcal N \setminus \{\a \cup \b\}$ and hence
\begin{equation}
\text{Pr}(\A_{\a}=0, \A_{\b}=0) >0.
\end{equation}

Since $\a, \b$ are nonempty and distinct, there are two possible cases.
In the first case, $\a \cap \b $ is nonempty. In this case, it can be checked easily that  either $\a \setminus  \b$,  $\b \setminus \a$, or both must be nonempty.  In the  second case, $\a \cap \b = \emptyset$. Then 
clearly both $\a \setminus  \b$ and  $\b \setminus \a$ must be nonempty. 
Finally, since $A$ has strictly positive probability distribution, then we can easily check that the theorem holds. 
\end{proof}

\begin{proposition}[Property 2: Subset]\label{prop:secAux6}
Suppose $\a,\b \in \PN$. Then 
\begin{align}\label{prop2:a}
H(\A_{\a}|\A_i, i\in\b ) >0
\end{align}
if and only if $\a \setminus \b $ is nonempty. 
\end{proposition}

\begin{proof}
By direct verification. 
\end{proof}

\begin{proposition}[Property 3: Partition]\label{prop:partition}
For any  $\a \in \PN$,  there exists random variables $\A_{\b_1},\ldots,\A_{\b_{n-2}}$  such that 
\begin{align}\label{prop3:a}
H(\A_{\b_k}|\A_{\a},\A_{\b_1},\ldots,\A_{\b_{k-1}})&>0  
\end{align}
for all $k=1, \ldots, n-2$.
\end{proposition}

\begin{proof}
Assume without loss of generality that $\a=\{i,\ldots,n\}$. Let
\begin{align}
\A_{\b_1}=\A_2, \A_{\b_2}=\A_3, \ldots,\A_{\b_{n-2}}=\A_{n-1}.
\end{align}
We can verified directly that \eqref{prop3:a} is satisfied. 
\end{proof}

In the rest of the paper,  we will assume without loss of generality that 
\begin{align}
p_{\Ao}(1) \ge p_{\Ao}(2) \ge \cdots \ge p_{\Ao}(n) > 0.
\end{align}


\begin{proposition}[Property 4: The smallest atom]\label{prop:smallestatom}
$\A_{n}$  has the minimum entropy among all $\a \in \PN$. In other words
\begin{align}\label{prop4:a}
\min_{\a \in \PN} H(\A_{\a}) = H(\A_{n}).
\end{align}
\end{proposition} 
\begin{proof}
Consider $\a \in \PN$. First notice that $p(n) \le p(i)$ for all $i = 1, \ldots, n$ and hence
$
\sum_{i\in\a} p_{\Ao}(i) \ge p_{\Ao}(n).
$
On the other hand, 
\begin{align}
\sum_{i\in\a} p_{\Ao}(i) \le \sum_{i\in\a} p_{\Ao}(i) + p_{\Ao}(1) - p_{\Ao}(n)   \le 1 - p_{\Ao}(n).
\end{align}
Therefore, 
$
p_{\Ao}(n) \le \sum_{i\in\a} p_{\Ao}(i) \le  1 - p_{\Ao}(n)
$ and  consequently $H(\A_{\a}) \ge H(\A_{n})$. The proposition is thus proved.
\end{proof}


\begin{proposition}[Property 5: Singleton $X_{i}$]\label{prop:smallestatomb}
Suppose $n \ge 2$. Then for any  $2 \leq i \leq n - 1$,
\begin{align}
H(\A_{i}|\A_{i+1},\ldots,\A_n) &>0.\label{prop5:a}
\end{align}
In addition, for all $\a\in\PN$ such that   $H(\A_{\a}|\A_{i+1},\ldots,\A_n) >0$, 
\begin{align}
H(\A_{i}|\A_{i+1},\ldots,\A_n) &\leq H(\A_{\a}|\A_{i+1},\ldots,\A_n)\label{prop5:b}\\
H(\A_{i}) &\leq H(\A_{\a}). \label{prop5:c}
\end{align}

\end{proposition}

\begin{proof}
Inequality \eqref{prop5:a} can be directly verified.  
Now, suppose $H(\A_{\a}|\A_{i+1},\ldots,\A_n) >0$. 
To prove \eqref{prop5:b} and \eqref{prop5:c}, first notice that 
\begin{align*}
& \lefteqn{H(\A_{\a}|\A_{i+1},\ldots,\A_n)}\nonumber \\
& \qquad =\sum_{a_{i+1}, \ldots, a_{n} \in \{0,1\}}p(\A_i=a_{i+1},\ldots,\A_n=a_{n})H(\A_{\a}|\A_{i+1}=a_{i+1},\ldots,\A_n=a_{n})\\
& \qquad \overset{(a)}{=} p(\A_{i+1}=0,\ldots,\A_n=0)H(\A_{\a}|\A_{i+1}=0,\ldots,\A_n=0).
%
\end{align*}
Here, $(a)$ follows from the fact that $\A = j$ if  $\A_{j}=1$.

Consider any $\a \in \PN$ which is not a subset of $\{i+1, \ldots, n \}$. Let 
$\iup \triangleq \{i + 1, \ldots, n \}$, and 
\begin{align}
q_{\a} \triangleq  \frac{\sum_{k\in \a \setminus \iup} p_{\Ao}(k)}{\sum_{j \not\in \iup} p_{\Ao}(j)}.
\end{align}
It can be verified easily that 
\begin{align}\label{eq:15}
H(\A_{\a}|\A_i=0,\ldots,\A_n=0) = h_{b}(q_{a})
\end{align}

As $\a$ is not a subset of $\iup$, there exists $k = 2, \ldots, i$ such that $k\in \a$. In this case,  
\begin{align}\label{eq:prop5c}
\sum_{j \in \a \setminus \iup} p_{\Ao}(j) \ge p_{\Ao} (k) \ge  p_{\Ao}(i).
\end{align}
Hence, 
$
q_{\a}   \ge  q_{i}.
$
On the other hand, 
\begin{align}\label{eq:prop5d}
 \sum_{j:j\in\a \setminus \iup} p_{\Ao}(j) \le  
 \sum_{j: j\in\a\setminus \iup} p_{\Ao}(j) +   p_{\Ao}(1)  - p_{\Ao}(i)  \le \sum_{j:j \not\in \iup} p_{\Ao}(j) - p_{\Ao}(i).
\end{align}
Hence, by dividing both \eqref{eq:prop5c} and \eqref{eq:prop5d} with $\sum_{j \not\in \iup} p_{\Ao}(j)$, we can prove that 
%
\[
1 - q(i) \ge q_{a}  \ge q_{i},
\]
and thus \eqref{prop5:b} holds. Now it remains to prove \eqref{prop5:c}.

Notice again from \eqref{eq:prop5c} and \eqref{eq:prop5d}  that 
%
%
%
%
$p_{\Ao}(i) \le \sum_{j \in \a  } p_{\Ao}(j) $ and 
\begin{align}
\sum_{j:j\in\a } p_{\Ao}(j) \le \sum_{j: j\in\a } p_{\Ao}(j) + p_{\Ao}(1) 
- p_{\Ao}(i) \le 1 - p_{\Ao}(i).
\end{align}
Consequently, 
$
H(\A_{i}) \le H(\A_{\a}) 
$ and the proposition is proved.
\end{proof}


\subsubsection{Uniqueness}

In the previous subsection, we have defined how to construct a set of auxiliary random variables from $\Ao$, and have identified properties of these random variables in relation to the underlying probability distributions. In the following, we will show that the constructed set of auxiliary random variables are in fact sufficient in fully characterizing the underlying probability distribution of $\Ao$. 


Let $\B$ be a random variable such that there exists auxiliary random   variables
\begin{align}
\{ \B_{a}, a\in \PN \}
\end{align}
such that 
\begin{align}
H(\A_{\a} , \a \in \alpha ) & = H(\B_{\a} , \a \in \alpha ), \quad \forall \alpha  \subseteq \PN  \label{eq:uniquenessa}\\
H(\A, \A_{\a} , \a \in \alpha ) & = H(\B, \B_{\a} , \a \in \alpha ), \quad \forall \alpha  \subseteq \PN.\label{eq:uniquenessb}
\end{align}
In other word, $\B$ is a random variable such that there exists auxiliary random variables $\{\B_{\a}, \a \in \PN\}$ such that the entropy function of 
$\{\B, \B_{\a}, \a \in \PN\}$ is essentially the same as that of $\{\Ao, \A_{\a}, \a \in \PN\}$. However, besides knowing the entropies of any subsets of random variables in $\{\B, \B_{\a}, \a \in \PN\}$, it is unknown how the random variables are constructed. Therefore, we cannot deduce immediate that $h_{b}(p_{\B}(n)) = H(\B_{n})$.
Yet, in the following, we will show that the entropy function itself is sufficient to characterize the underlying distribution $\B$, under a mild condition that the size of the sample space of $\B$ is no larger than that of  $\A$.

To achieve our goal, we will need to prove a few intermediate results. 
We will assume the condition that  $\support(\B) \le n$ is satisfied.
Also, for simplicity, we will assume without loss of generality that  $\B$ is a random variables such that 
\begin{align}
p_{\B}(1) \ge p_{\B}(2)    \ge \cdots, \ge p_{\B}(n).
\end{align}

 
\begin{theorem}[Binary random variables] \label{thm1}
Each $\B_{\a}$ is binary. In other words, 
\begin{align}
\support(\B_{\a}) = 2, \quad \forall \a \in \PN.
\end{align}
Also, all the random variables in $\{ \B_{\a}, \a \in \PN\}$ are distinct.
\end{theorem}
\begin{proof}
By \eqref{prop1:a}, \eqref{prop1:b}, \eqref{eq:uniquenessa}, \eqref{eq:uniquenessb}, we can conclude right away that 
all $\B_{a}$ are distinct. It remains to prove that each $\B_{a}$ is indeed binary. 
To see this, first notice that  Proposition \ref{prop:partition} implies that for  any  $\a \in \PN$, there exists $\b_{1} , \ldots , \b_{n-2} \in \PN$ such that   
\begin{align}
H(\A_{\b_k}|\A_{\a},\A_{\b_1},\ldots,\A_{\b_{k-1}})&>0  
\end{align}
for all $k=1, \ldots, n-2$. Hence, by \eqref{eq:uniquenessa} and \eqref{eq:uniquenessb}, 
\begin{align}
H(\B_{\b_k}|\B_{\a},\B_{\b_1},\ldots,\B_{\b_{k-1}})&>0  
\end{align}
for all $k=1, \ldots, n-2$.

Notice that for any random variables $C'$ and $C''$, 
\begin{align}
 H(C', C'')>H(C')\Leftrightarrow \support(C', C'')>\support(C').
\end{align}
Consequently, we prove that 
\begin{align}\label{eq:23}
2\le \support(Y_{\a}) < \support(Y_{\b_{1}}, Y_{\a}) < \ldots <\mathcal \support(Y_{\a}, Y_{\b_{1}} ,\ldots,  Y_{\b_{n-2}}).
\end{align}
Since $Y_{\a}, Y_{\b_{1}} ,\ldots,  Y_{\b_{n-2}}$ are functions of $\B$,  we have 
\begin{align}
\support(Y_{\a}, Y_{\b_{1}} ,\ldots,  Y_{\b_{n-2}}) \leq n. 
\end{align}
Together with \eqref{eq:23}, this implies that 
$\support(Y_{\a}) = 2$. In other words, $\B_{\a}$ is binary. 
\end{proof}

\begin{corollary}\label{cor:1}
$\B$ is positive. In other words, $\support(\B) = n$.
\end{corollary}
\begin{proof}
If $\support(\B) \le n$, then there are at most $2^{n-1} - 1$ distinct binary random variables. Since $|\PN| = 2^{n-1} - 1$,  we can conclude that $\support(\B) = n$.
\end{proof}

\begin{corollary}\label{cor:2}
For any $\a \in \PN$, there exists $\b \in \PN$ such that
\begin{align}
\B_{\a} = \B^{*}_{\b}
\end{align}
where 
\begin{equation}
  \B^{*}_{\b} = \left\{
  \begin{array}{l l}
    1 & \quad \text{if $\B \in \b$}\\
    0 & \quad \text{otherwise}\\
  \end{array} \right.
\end{equation}

Conversely, for any $\b \in \PN$, there exists $\a \in \PN$ such that 
\begin{align}
\B_{\a} = \B^{*}_{\b}.
\end{align}
Equivalently, there exists a one-to-one mapping $\sigma$ between $\PN$ such that 
\begin{align}
\B_{\a} = \B^{*}_{\sigma(\a)}.
\end{align}
\end{corollary}
\begin{proof}
The corollary follows directly from the fact that  there are exactly $|\PN|$'s distinct binary variables which are functions of $\B$.
\end{proof}

Among all the auxiliary random variables $\{\B_{\a}^{*}, \a \in \PN\}$, 
 the random variables $\B^{*}_{1^{+}}$ and $\B^{*}_{i}$  for $i=2, \ldots, n$ are called ``\emph{indicator}'', to highlight that these variables can derived from an indicator function of $\B$. In particular,  
\begin{align}
\B^{*}_{i} = 
\begin{cases}
1 & \text{ if } \B = i \\
0 & \text{ otherwise,}
\end{cases}
\end{align}
and
\begin{align}
\B^{*}_{1^{+}} = 
\begin{cases}
0 & \text{ if } \B = 1 \\
1 & \text{ otherwise.}
\end{cases}
\end{align}

{\bf Remark: } Note that when $n=2$, there is only one indicator random variable because $\B^{*}_{1^{+}}$ and $\B^{*}_{2}$ are equivalent (meaning that they are functions of each other).

{\bf Remark: }
For $i=2, \ldots, n$ where $n\ge 3$,   the ``\emph{atomic probability}'' of the indicator random variable $\B^{*}_{i}$ is defined as  $p_{\B}(i)$. Similarly, 
the atomic probability of  $\B^{*}_{1^{+}}$ is defined as  $p_{\B}(1)$. In the special case when $n=2$ where $\B^{*}_{2}$ and $\B^{*}_{1^{+}}$ are equivalent, the atomic probability of $\B^{*}_{2}$ and $\B^{*}_{1^{+}}$ will conventionally be defined as 
$p_{\B}(2)$.

While it is explicitly defined what indicator random variables are, it is not explicitly clear which random variable $\B_{\a}$ is an indicator random variable as $\sigma$ is unknown. In the following, by exploiting properties of the entropy function of the auxiliary random variables, we can in fact identify if $\B_{\a}$ is   an indicator random variable or not.

\begin{theorem}\label{thm2}
$\B_{n}$ is an indicator such that  
\begin{align}
H(\B_{n}) = H(\B^{*}_{n}).
\end{align}
In other words, its atomic probability  is equal to  $h_{b}^{-1}(H(Y_{n}))$.
\end{theorem}
\begin{proof}
In Theorem \ref{thm1}, we proved that the auxiliary random variables 
$\{\B_{a}, a\in \PN\}$ are all distinct, and binary.
By \eqref{prop4:a},  \eqref{eq:uniquenessa}, \eqref{eq:uniquenessb}, we can deduce  that 
\begin{align} 
\min_{\a \in \PN} H(\B_{\a}) = H(\B_{n}).
\end{align}

At the same time, by Proposition \ref{prop:smallestatom},  $\min_{\a \in \PN} H(\B^{*}_{\a}) = H(\B^{*}_{n})$, we can thus conclude right away that 
$H(\B_{n}) = H(\B^{*}_{n}) = h_{b}( p_{\B}(n))$. However, it is not sufficient to conclude immediately
that $\B_{n}  = \B^{*}_{n}$ and hence an indicator random variable. In the following, we will show that this is indeed the case.

By  Corollary \ref{cor:2},  there exists $\a \in \PN$ such that 
\begin{align}
\B_{n} = \B^{*}_{\a}.
\end{align}
As $ H(\B^{*}_{\a} ) = H( \B^{*}_{n}) $,  we must have 
\begin{align}\label{eq:thm2a}
p(n) = \sum_{i\in\a} p_{\B}(i)  
\end{align}
or 
\begin{align}\label{eq:thm2b}
\sum_{i\in\a} p_{\B}(i)  = 1- p(n).
\end{align}

Suppose the case \eqref{eq:thm2a} holds. Since $p_{\B}(i)   \ge p_{\B}(n)$ for all $i$, 
$|\a|  =  1$. Also, $p_{\B}(\a) = p(n)$. 
Now, suppose the other case \eqref{eq:thm2b} holds. Thus, 
\begin{align}
\sum_{i\in\a} p_{\B}(i)  + p_{\B}(n) = 1.
\end{align}
However, as $\a \subseteq \{2, \ldots, n\}$ and 
$
\sum_{i=1}^{n} p_{\B}(i) = 1
$,
we can immediately see   that 
$i$) $p_{\B}(1) =  p_{\B}(n) $ and hence 
$p_{\B}(1) = \cdots = p_{\B}(n) = {1}/{n}$, and $ii$) $\a = \{2, \ldots, n\}$. 
Again, it implies that $\B_{n}$ is a singleton. In addition, the atomic probability of the indicator random variable  is equal to $p_{\B}(n)$. Therefore, in both cases, we proved that $\B^{*}_{a} = \B_{n}$ is an indicator random variable, with atomic probability equal to $p_{\B}(n)$.
\end{proof}

By Theorem \ref{thm2},  the atomic probability of the indicator random variable  $\B_{n}$ is the smallest. Therefore, we may  assume\footnote{
This can be achieved by properly renaming the alphabets of $\B$.
} without loss of generality that 
$\B_{n}$ is in fact equal to $\B^{*}_{n}$.
In the following, we will prove by induction that 
$\B_{i}$ can in fact be viewed as   $\B^{*}_{i}$, again after properly renaming 
alphabets of $\B$.


\begin{theorem}
Let $i \ge 2$. Suppose  
\begin{align}\label{eq:thm3a}
\B_{j} = \B^{*}_{j}, \: \forall j= i+1, \ldots, n.
\end{align}
Then $\B_{i}$ is an indicator random variable  with
 atomic probability of $\B_{i}$  equal to $p_{\B}(i)$. 
 
 When $n \ge 3$, 
 $\B_{1^{+}}$ is also an indicator random variable with atomic probability equal to $p_{\B}(1)$.
\end{theorem}
\begin{proof}
First, recall that 
$\{\B_{\a} , \a \in \PN \}$ and $\{\B^{*}_{\a} , \a \in \PN \}$ are two identical sets of random variables.
By \eqref{prop5:b}-\eqref{prop5:c} and  \eqref{eq:uniquenessa}, we have $H(\B_{i} | \B_{i+1},\ldots,\B_n)  > 0$. 
On the other hand, $H(\B^{*}_{i} | \B_{i+1},\ldots,\B_n) > 0 $ by \eqref{prop:smallestatomb} and \eqref{eq:thm3a}. 
By Proposition \ref{prop:smallestatomb} and also \eqref{eq:uniquenessa}, 
\begin{align*}
H(\B_{i}|\B_{i+1},\ldots,\B_n) \le  H(\B^{*}_{i} | \B_{i+1},\ldots,\B_n) \le H(\B_{i}|\B_{i+1},\ldots,\B_n)
\end{align*}
and 
\begin{align*}
H(\B_{i}) \le  H(\B_{i}^{*}) \le H(\B_{i}). 
\end{align*}
%
Thus, we proved that 
\begin{align}
H(\B_{i}|\B_{i+1}^{*},\ldots,\B_n^{*}) & = H(\B^{*}_{i}|\B_{i+1} ,\ldots,\B_n ) \label{eq:35}\\
H(\B_{i}) &=  H(\B^{*}_{i}).  \label{eq:36}
\end{align}

As in \eqref{eq:15}, we can show that 
\begin{align}
H(\B^{*}_{\a}|\B^{*}_{i+1},\ldots,\B^{*}_n)    & = p(\B^{*}_{i+1}=0,\ldots,\B^{*}_n=0)H(\B^{*}_{\a}|\B^{*}_{i+1}=0,\ldots,\B^{*}_n=0).
\end{align}

Let 
\begin{align}
q_{\a} \triangleq  \frac{\sum_{k\in \a \setminus \iup} p_{\B}(k)}{\sum_{j \not\in \iup} p_{\B}(j)}.
\end{align}
Then,  
$
H(\B^{*}_{\a}|\B^{*}_i=0,\ldots,\B^{*}_n=0) = h_{b}(q_{a})
$.

By \eqref{eq:35}, we have  
$
h_{b}(q_{a}) = h_{b}(q_{i})
$, or equivalently,  
\begin{align}\label{eq:37}
\frac{\sum_{k\in \a \setminus \iup} p_{\B}(k)}{\sum_{j \not\in \iup} p_{\B}(j)} = \frac{ p_{\B}(i) }{\sum_{j \not\in \iup} p_{\B}(j)}.
\end{align}
or 
\begin{align}\label{eq:38}
\frac{\sum_{k\in \a \setminus \iup} p_{\B}(k)}{\sum_{j \not\in \iup} p_{\B}(j)} = 1 - \frac{ p_{\B}(i) }{\sum_{j \not\in \iup} p_{\B}(j)}.
\end{align}

Suppose \eqref{eq:37} holds. As $p_{\B}(k) \ge p_{\B}(i)$ for all $k \in \a \setminus \iup$, we can conclude right away that  $\a \setminus \iup$ contains only one element whose probability is equal to $p_{\B}(i)$. 
In addition, due to \eqref{eq:36}, either $\a \cap \iup$ is empty or
\begin{align}
\sum_{k\in \a  } p_{\B}(k) = 1 - p_{\B}(i). 
\end{align}
 If $\a \cap \iup$ is empty, then $\a$ has only one element and the theorem holds. 
 Now, suppose $\sum_{k\in \a  } p_{\B}(k) = 1 - p_{\B}(i)$. 
 This must imply that 
 \begin{align}
 p_{\B}(1) = p_{\B}(i)
 \end{align}
 and $\a = \{2, \ldots, n \}$.
In this case, of course $\B_{i}$ is still an indicator random variable, with its atomic probability equal to $p_{\B}(i)$.

Now, suppose on the other hand that   \eqref{eq:38} holds. Then 
 \begin{align}
{ p_{\B}(i) } +  {\sum_{k\in \a \setminus \iup} p_{\B}(k)} = {\sum_{j \not\in \iup} p_{\B}(j)} . 
 \end{align}
This implies that 1) $ p_{\B}(1) =  p_{\B}(i)$ and 2) 
$\a \setminus \iup = \{2, \ldots, i\}$. Now, by \eqref{eq:36}, either
\begin{align}
\sum_{k\in \a  } p_{\B}(k) = 1 - p_{\B}(i). 
\end{align}
This further implies that $\a = \{2, \ldots, n\}$. Hence, again $\B_{i}$ is an indicator random variable with atomic probability equal to $p_{\B}(i)$.

In all the cases,  we have proved that for $i=2, \ldots, n$, 
$\B_{i}$ is an indicator random variable whose atomic probability is $p_{\B}(i)$.
When $n=2$, there is only one indicator random variable. But when $n \ge 3$, there will be $n$ such indicator random variables and $\B_{1^{+}}$ is also an indicator random variable. 

To see this, first recall Proposition \ref{prop:secAux6} that for all $\a,\b \in \PN$,
\begin{align}
H(\B^{*}_{\a} | \B^{*}_{i}, i\in\b) = 0
\end{align}
if and only if $\a \subseteq \b$. Also, 
$
H(\B_{1^{+}} | \B_{i }, i \in\a) > 0 
$
for all $\a \neq 1^{+}$. Sincc $\B_{i} = \B^{*}_{i}$ for $i=2, \ldots, n$, it implies that  
$\B_{1^{+}} = \B^{*}_{1^{+}}$. 
\end{proof}

\begin{theorem}\label{thm:singlemain}
Suppose $\Ao$ is an $n$-ary random variable taking values from the set 
$\{1, \ldots, n\}$,  and has a positive probability distribution. 
For any $\a \subseteq \PN$, let $\A_{\a}$ be the auxiliary random variable such that 
\begin{equation}
  \A_{\a} = \left\{
  \begin{array}{l l}
    1 & \quad \text{if $X \in \a$}\\
    0 & \quad \text{otherwise.} 
  \end{array} \right.
\end{equation}

Let $\B$ be another $n$-ary random variable such one can construct a set of  random variables $\{\B_{\a},\a \subseteq  \PN\}$ such that 
\begin{align}
H(\B_{\a }, \a \in\alpha) & = H(\A_{a},a\in\alpha), \quad \forall \alpha \subseteq \PN \\
H( \B_{a} | \B) & =0, \quad \forall a \in \PN.
\end{align}
Then the probability distributions of $\Ao$ and $\B$ are permutations of each other. 
\end{theorem}
\begin{proof}
In the previous results, we proved that 
$\B_{i}, \ldots, \B_{n}$ and $\B_{1^{+}}$ are in fact  all indicator random variables such that 
$\B_{i} = \B^{*}_{i}$ for $i=2, \ldots, n$ (up to permutations). Hence, 
\begin{align}
p_{\B}(i) = h_{b}^{-1}( H(\B_{i}) )
\end{align}
for $i=2, \ldots, n$. As $H(\A_{i})  = H(\B_{i}) $, $p_{\B}(i) = p_{\Ao}(i)$ for all $i$. The theorem is proved.
\end{proof}

\subsection{Extension for Multiple Random Variables}
In the following, we will further extend Theorem \ref{thm:singlemain}  to multiple random variable cases.

Suppose $\Ao=(\Ao_{1},\ldots, \Ao_{m})$ is a set of random variables defined over the sample space $\prod_{i=1}^{m}\X_{i} $ such that 
\begin{enumerate}
\item
$|\X_{i}| \ge 3$ for all $i=1,\ldots, m$; 

\item 
the probability distribution of $(\Ao_{1},\ldots, \Ao_{m})$ is positive.
\end{enumerate}

Let  ${\bf 1} \in \prod_{i=1}^{n}\X_{i} $ and $\Omega \triangleq  \prod_{i=1}^{n}\X_{i} \setminus \{ {\bf 1} \}$. For any $\a  \subseteq \Omega$, let 
\begin{equation}
  \A_{\a} = \left\{
  \begin{array}{l l}
    1 & \quad \text{if $(X_{1},\ldots, X_{m}) \in \a$}\\
    0 & \quad \text{otherwise.}
  \end{array} \right.
\end{equation}

Now, consider a  set of random variables  $\B = (\B_{1}, \ldots, \B_{m})$  defined over the same sample space $\prod_{i=1}^{n}\X_{i} $. Suppose that there exists random variables $\{B_{\a}, \a \in \Omega\}$ such that 
\begin{align}
H(B_{\a}, \a \in \alpha) & = H(\A_{\a}, \a \in \alpha), \quad \forall \alpha \subseteq   {\mathsf P} (\Omega ) \\
H(B_{a} | \B) &  = H(\A_{a} | \Ao) = 0, \quad \forall \a \in \Omega.
\end{align}
Then the probability distributions of $\Ao=(\Ao_{1},\ldots, \Ao_{m})$ and $\B = (\B_{1}, \ldots, \B_{m})$ are permutations of each other. Or more precisely, there exists permutations $\sigma_{i}$ on the set $\X_{i}$ for $i=1,\ldots, m$ such that the probability distributions of $(X_{1},\ldots, X_{m})$ and 
$(\sigma_{1}(\B_{1}), \ldots, \sigma_{n}(\B_{n}))$ are exactly the same.

%
%
%
%
%
%


\section{Conclusion}
In this paper we showed that the known outer bounds are not tight and can be improved by introducing auxiliary random variables. This led us to new improved bounds $\mathcal R'_{\text{cs}} (\overline{\Gamma^*})$ and $\mathcal R'_{\text{cs}}(\Gamma)$. The interrelation of the bounds presented in this paper is summarized as follows.
\begin{equation*}
\mathcal R_{\text{cs}} \subset
\left\{
\begin{array}{c}
\mathcal R'_{\text{cs}}(\overline{\Gamma^*}) \subset \mathcal R_{\text{cs}}(\overline{\Gamma^*})\\
\mathcal R'_{\text{cs}}(\Gamma)
\end{array} \right\}\subset \mathcal R_{\text{cs}}(\Gamma) \subset \mathcal R_{FD}
\end{equation*}
We also discussed construction of auxiliary random variables to tighten the bounds. We established that theoretically it is feasible to incorporate the knowledge of probability distribution of random variables completely in the entropy domain by construction of auxiliary random variables. However, the construction of auxiliary random variables describing the exact probability distribution of a given random variables in the proposition is very complex. Hence, it remains an interesting practical problem to find simpler construction methods for auxiliary random variables.

\bibliographystyle{ieeetr}
\bibliography{network}

\begin{thebibliography}{10}

\bibitem{ChaGra08}
T.~H. Chan and A.~Grant, ``Dualities between entropy functions and network
  codes,'' {\em IEEE Trans. Inform. Theory}, vol.~54, pp.~4470--4487, Oct.
  2008.

\bibitem{Yeu02}
R.~W. Yeung, {\em A First Course in Information Theory}.
\newblock New York: Kluwer Academic/Plenum Publishers, 2002.

\bibitem{ThaChaGra11}
S.~Thakor, T.~Chan, and A.~Grant, ``Bounds for network information flow with
  correlated sources,'' in {\em Australian Communications Theory Workshop
  (AusCTW)}, (Melbourne, Australia), pp.~43 --48, Feb. 2011.

\bibitem{ThaGraCha12}
S.~Thakor, A.~J. Grant, and T.~Chan, ``Cut-set bounds on network information
  flow,'' {\em CoRR}, vol.~abs/1305.3706, 2013.

\bibitem{AhlCai00}
R.~Ahlswede, N.~Cai, S.-Y.~R. Li, and R.~W. Yeung, ``Network information
  flow,'' {\em IEEE Trans. Inform. Theory}, vol.~46, pp.~1204--1216, July 2000.

\bibitem{YeuZha99}
R.~W. Yeung and Z.~Zhang, ``On symmetrical multilevel diversity coding,'' {\em
  IEEE Trans. Inform. Theory}, vol.~45, pp.~609--621, Mar. 1999.

\bibitem{Yeu08}
R.~W. Yeung, {\em Information Theory and Network Coding}.
\newblock Springer, 2008.

\bibitem{LiYeu03}
S.-Y.~R. Li, R.~Yeung, and N.~Cai, ``Linear network coding,'' {\em IEEE Trans.
  Inform. Theory}, vol.~49, pp.~371--381, Feb. 2003.

\bibitem{RamJaiChoEff06}
A.~Ramamoorthy, K.~Jain, P.~Chou, and M.~Effros, ``Separating distributed
  source coding from network coding,'' {\em IEEE Trans. Inform. Theory},
  vol.~52, pp.~2785 -- 2795, Jun. 2006.

\bibitem{Han11}
T.~S. Han, ``Multicasting multiple correlated sources to multiple sinks over a
  noisy channel network,'' {\em IEEE Trans. Inform. Theory}, vol.~57, pp.~4
  --13, Jan. 2011.

\bibitem{Han80}
T.~S. Han, ``Slepian-{W}olf-{C}over theorem for a network of channels,'' {\em
  Inform. Control}, vol.~47, no.~1, pp.~67--83, 1980.

\bibitem{BarSer06}
J.~Barros and S.~Servetto, ``Network information flow with correlated
  sources,'' {\em IEEE Trans. Inform. Theory}, vol.~52, pp.~155 -- 170, Jan.
  2006.

\bibitem{ThaGraCha09}
S.~Thakor, A.~Grant, and T.~Chan, ``Network coding capacity: A functional
  dependence bound,'' in {\em Information Theory, 2009. ISIT 2009. IEEE
  International Symposium on}, pp.~263 --267, 28 2009-july 3 2009.

\bibitem{ZhaYeu97}
Z.~Zhang and R.~Yeung, ``A non-shannon-type conditional inequality of
  information quantities,'' {\em IEEE Trans. Inform. Theory}, vol.~43, pp.~1982
  --1986, Nov. 1997.

\bibitem{GohYanJagg11}
A.~A. Gohari, S.~Yang, and S.~Jaggi, ``Beyond the cut set bound: Uncertainty
  computations in network coding with correlated sources,'' {\em CoRR},
  vol.~abs/1103.3596, 2011.

\bibitem{GacKor73}
P.~G\'{a}cs and J.~Korner, ``{Common information is far less than mutual
  information},'' {\em Probl. Inform. Control}, vol.~2, no.~2, pp.~149--162,
  1973.

\bibitem{Wyn73}
A.~Wyner, ``A theorem on the entropy of certain binary sequences and
  applications--ii,'' {\em Information Theory, IEEE Transactions on}, vol.~19,
  no.~6, pp.~772--777, 1973.

\end{thebibliography}

\end{document}